\pdfoutput=1
\documentclass[journal,onecolumn,10pt]{IEEEtran}

\usepackage{amssymb}
\usepackage{latexsym}
\usepackage{graphicx}
\usepackage{amsmath}
\usepackage{bm}
\usepackage{hyperref}
\usepackage{amsthm}
\usepackage{color}
\usepackage{blkarray}
\usepackage{algorithm}
\usepackage{algorithmic}
\numberwithin{equation}{section}
\usepackage{tikz} 
\usetikzlibrary{shapes.geometric, arrows} 

\newtheorem{thm}{Theorem}[section]
\newtheorem{lemma}[thm]{Lemma}
\newtheorem{pro}[thm]{Proposition}

\newtheorem{defi}[thm]{Definition}

\newtheorem{rem}[thm]{Remark}
\newtheorem{Exa}{Example}[section]

\newcommand{\x}{\bm{x}}
\newcommand{\y}{\bm{y}}
\newcommand{\m}{\bm{m}}

\newcommand{\bmA}{\bm{A}}
\newcommand{\bmB}{\bm{B}}
\newcommand{\bmI}{\bm{I}}
\newcommand{\bmH}{\bm{H}}
\newcommand{\bmD}{\bm{D}}
\newcommand{\bmL}{\bm{L}}

\newcommand{\F}{\mathbb{F}}

\DeclareMathOperator{\Ker}{Ker}

\ifCLASSINFOpdf
\else
\fi
\begin{document}
%
\title{An Open Problem on Sparse Representations in Unions of Bases}
%
%
%

\author{
	Yi Shen, Chenyun Yu, Yuan Shen and Song Li
\thanks{
	Yi Shen, Chenyun Yu, Yuan Shen are with
	Department of Mathematics, Zhejiang Sci--Tech University, Hangzhou 310018, China
}
\thanks{Corresponding Author: Song Li is with School of Mathematical Science, Zhejiang  University, Hangzhou 310027, China}
}

\maketitle

\begin{abstract}
    We consider sparse representations of signals from redundant dictionaries which are unions of several orthonormal bases.
    The spark introduced by Donoho and Elad plays an important role in sparse representations.
    However,  numerical computations of  sparks are generally combinatorial.
    For unions of several orthonormal bases,
    two lower bounds on the spark  via the mutual coherence were established in previous work.
    We constructively prove that both of them are tight.
    Our main results give  positive answers to Gribonval and Nielsen's open problem on sparse representations in unions of orthonormal bases.
    Constructive proofs rely on a family of mutually unbiased bases which first appears in quantum information theory.
\end{abstract}

\begin{IEEEkeywords}
    Sparse Approximation,
    Spark,
    Mutual Coherence,
    Latin Squares,
    Mutually Unbiased Bases.
\end{IEEEkeywords}


\section{Introduction}
Given a redundant dictionary,
the problem of representing  vectors (also referred to as signals)  with   linear combinations of  small numbers of atoms
from the dictionary  is called  the \textit{sparse representation} \cite{Bruckstein,2012Compressed,2013foucart}.
Two fundamental concepts  defined in \cite{Bruckstein,donoho2003} are  core issues of sparse representations.
One is the  mutual coherence, the other is the spark. The matrix notation $\bm{D}$ is used for a dictionary. The \textit{spark} denoted by $\eta(\bm{D})$ is defined to be the smallest number of columns from matrix $\bm{D}$ that are linearly dependent.  The \textit{mutual coherence} denoted by $\mu(\bm{D})$ is defined to be the largest absolute normalized inner product between different columns from the matrix $\bm{D}$.
The value of spark is difficult to evaluate, the mutual coherence is used to estimate the spark in \cite{donoho2003,elad2002,Gribonval2003}.
Our interest in this paper centers around tightness of two lower bounds for the spark obtained in the previous work.

\subsection{Background}
By $\|\x\|_0$ we refer to the number of nonzero entries of  the vector $\x$.
A vector is said to be  $q$ \textit{sparse} if $\|\x\|_0\le q$.
By definition of  the spark, we see that
$$
\eta(\bm{D}) = \min_{\x\in \Ker(\bm{D}),\ \x\ne 0} \|\x\|_0
$$
where $\Ker(\bm{D})$ denotes the null space of $\bm{D}$.
The spark is useful to bound the sparsity of the uniqueness of sparse solutions.
If a  linear system $\y =\bm{D}\x$ has a solution
$\x$ obeying $\|\x\|_0 < \eta(\bm{D})/2$,
then  this solution is necessarily the sparsest
possible \cite{donoho2003}. Therefore, large values of spark are expected in applications.
To estimate the spark, lower bounds depending on the mutual coherence were obtained in \cite{donoho2003,elad2002,Gribonval2003}.

For any given arbitrary dictionary $\bm{D}$, Elad and Donoho proved in \cite{donoho2003} that
\begin{equation}\label{eq:maing}
	\eta(\bm{D})  \ge 1 + \frac{1}{\mu (\bm{D)}}.
\end{equation}
If  $\bm{D}$ is assumed to be a union of two orthonormal bases, then
a tighter estimate
\begin{equation}\label{eq:maing2}
	\eta(\bm{D})  \ge  \frac{2}{\mu (\bm{D)}}
\end{equation}
was obtained by Elad and Bruckstein in \cite{elad2002}.
If the dictionaries are Dirac/Fourier matrix pairs,
then the inequality \eqref{eq:maing2} reduces to  the support uncertainty principle obtained in  \cite{donoho2001,donoho1989}.
Extensions of the support uncertainty principle to the Fourier transform on abelian groups are referred to \cite{Feng, karhmer2008,quisquater2003,Tao}
and references therein.
Robust uncertainty principles were proved to  hold for most supports in time and frequency \cite{2006Robust,2006Quantitative}.
The case that  dictionaries are concatenations of several orthonormal bases were  studied in \cite{donoho2003,donoho2001,Gribonval2003,Gribonval2007}.
Let $q$ denote a positive integer. If  dictionaries $\bm{D}$ are  unions of $q+1$ orthonormal bases,
then  Gribonval and Nielsen proved in \cite[Lemma 3]{Gribonval2003}  that
\begin{equation}\label{eq:main}
	\eta(\bm{D})  \ge \left( 1 + \frac{1}{q} \right)  \frac{1}{\mu(\bm{D})}.	
\end{equation}

\subsection{Motivations}

The inequality \eqref{eq:main}  is a natural generalization of the inequality \eqref{eq:maing2}
from $q=1$ to $q \ge 2$.
The estimate  \eqref{eq:maing2} is tight, since there exist  Dirac/Fourier matrix pairs that meet the bound with equality \cite{donoho2003,donoho1989,elad2002}.
However,
the question whether the estimate \eqref{eq:main} for $q\ge 2$ is tight remains open.
This open question was further discussed  by Gribonval and Nielsen  in \cite{Gribonval2007}.
For general dictionary,
the tightness of  the estimate \eqref{eq:maing}  has been studied in \cite[Section III]{cai}.
For unions of several orthonormal bases, to the best of our knowledge, the tightness of  the estimate \eqref{eq:maing} is also unknown. Motivated by these open questions, we study  the tightness of the estimate \eqref{eq:maing}  and  the estimate \eqref{eq:main}.

\subsection{Observations}
Before going further,  we discuss some observations. Following the line in \cite{Gribonval2003,Gribonval2007}, we are concerned with unions of three or more orthonormal bases.
Suppose that the  number of orthonormal bases is $q+1$. We discuss two cases.
\begin{enumerate}
	\item
	Let the mutual coherence of $\bm{D}$  be $1/q$. Then
	\begin{itemize}
		\item The right-hand side of \eqref{eq:maing} is equal to $q+1$.
		\item The right-hand side of \eqref{eq:main} is equal to $q+1$.
		\item If $\x$ is in  $\Ker(\bm{D})$, then it follows from the bound \eqref{eq:main} that $\|\x\|_0 \ge q+1$.
	\end{itemize}
	Based on  those observations, for $q=2^m$, $m=1,2, 3, \ldots$, the goal is to find a dictionary $\bm{D}$
	and a corresponding sparse vector $\x$ that satisfy the following three conditions
	$$
	\bm{D}\x = \bm{0}, \quad \|\x\|_0 = q+1,\quad \mu(\bm{D})=\frac{1}{q}.
	$$
	\item Let the mutual coherence of $\bm{D}$  be $1/q^2$. Then
	\begin{itemize}
		\item The right-hand side of \eqref{eq:maing} is equal to $q^2+1$.
		\item The right-hand side of \eqref{eq:main} is equal to $q^2 + q$.
		\item If $\x$ is in $\Ker(\bm{D})$, then it follows from the bound \eqref{eq:main} that $\|\x\|_0 \ge q^2 + q$.	
	\end{itemize}	
	Based on  those observations, for $q=2^m$, $m=1, 2, 3, \ldots$, the goal is to find a dictionary $\bm{D}$
	and a corresponding sparse vector $\x$ that satisfy the following three conditions
	$$
	\bm{D}\x = \bm{0}, \quad \|\x\|_0 = q^2+q,\quad \mu(\bm{D})=\frac{1}{q^2}.
	$$
\end{enumerate}

\subsection{Contributions}

The main contributions of this paper  are summarized as follows.

\begin{thm}\label{thm:L0}
	For any $q=2^m$, $m=1,2,3,\ldots$, there exists a dictionary $\bm{D}$ which is a union of $q+1$ orthonormal bases satisfies
	$\eta(\bm{D}) = q+1$ and  $\mu(\bm{D})=1/q$. Then
	$$
	\eta(\bm{D})  = 1 + \frac{1}{\mu (\bm{D)}} \quad
	\mbox{and}\quad
	\eta(\bm{D})=
	\left( 1 + \frac{1}{q} \right)  \frac{1}{\mu(\bm{D})}.
	$$
\end{thm}

\begin{thm}\label{thm:L02}
	For any $q=2^m$, $m=1,2,3,\ldots$, there exists a dictionary $\bm{D}$ which is a union of $q+1$ orthonormal bases satisfies
		$\eta(\bm{D}) = q^2+q$ and  $\mu(\bm{D})=1/q^2$. Then
	$$
	\eta(\bm{D})  > 1 + \frac{1}{\mu (\bm{D)}}
	\quad
	\mbox{and}\quad
	\eta(\bm{D})=
	\left( 1 + \frac{1}{q} \right)  \frac{1}{\mu(\bm{D})}.
	$$
\end{thm}

Therefore,
for unions of  several orthonormal bases,
both the inequality \eqref{eq:maing} and the inequality \eqref{eq:main} are achievable.
Theorem \ref{thm:L02} implies that the estimate \eqref{eq:main} is sharper than the  estimate \eqref{eq:maing}
for some special dictionaries.
Now we can answer  Gribonval and Nielsen's open problem in \cite{Gribonval2003} positively:
\begin{center}
	``\textit{There  exist examples of $q+1$ orthonormal bases for which
		$
		\eta(\bm{D})\mu(\bm{D}) = 1 + {1}/{q}.
		$}''
\end{center}

\subsection{Flowchart}

Our constructive proofs base on techniques from discrete mathematics and quantum information theory \cite{Bose1938,colbourn1996,Thomas,Wocjan2004}.
Two families of dictionaries are constructed  by using the \textit{mutually unbiased bases} (MUBs) obtained in \cite{Wocjan2004}.
The process of construction is illustrated in Figure \ref{fig1} step by step.
Symbols' meanings are listed in Table \ref{tab:1}.
The  existence theorem of MUBs is clear, see e.g. \cite{Thomas,Wocjan2004}.
To find  specific sparse vectors in the null space of such dictionaries, however, more detailed properties need to be established.

\tikzset{
	nonterminal/.style = {
		rectangle,
		align = center,
		minimum size = 6mm,
		very thick,
		draw = red!50!black!50,
		top color = white,
		bottom color = red!50!black!20,
	}
}
\tikzset{
	terminal/.style = {
		rectangle,
		align = center,
		minimum size = 6mm,
		rounded corners = 3mm,
		very thick,
		draw = black!50,
		top color = white,
		bottom color = black!20,
	}
}
\tikzstyle{line} = [draw=purple, thick, -latex']

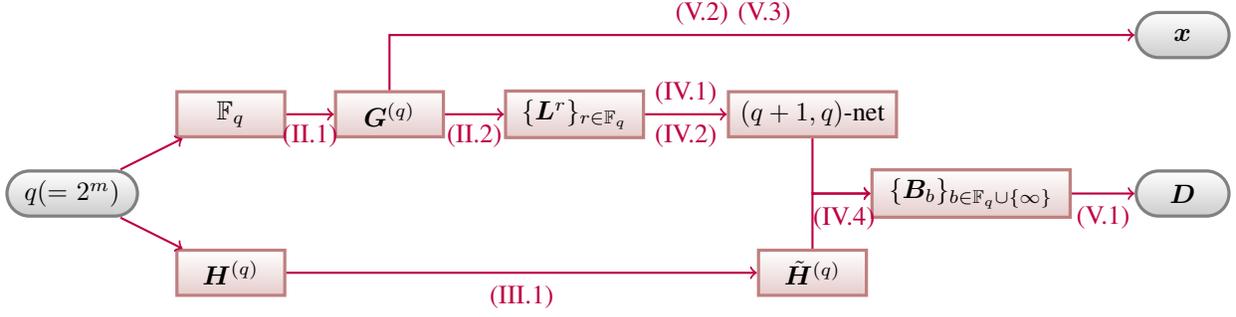
\begin{figure*}
	
	\begin{center}
		
		\begin{tikzpicture}[node distance = 3em, auto]
			
			\node [terminal, text width = 1.5cm] (init) {$q(=2^m)$};
			\node [nonterminal, text width = 1.2cm, right of= init, xshift= 3em,yshift= 3em](GF){$\mathbb{F}_q$};
			\node [nonterminal, text width = 1.2cm, right of= GF, xshift= 3em](Gs){$\bm{G}^{(q)}$};
			\node [nonterminal, text width = 1.6cm, right of= Gs, xshift= 4em](L){$\{\bm{L}^{r}\}_{r\in \F_q}$};
			\node [nonterminal, text width = 2cm, right of=  L, xshift= 6em](Net){$(q+1,q)$-net};
			\node [nonterminal, text width = 2.4cm, right of = Net,xshift= 3em,yshift= -3em](B){$\{\bm{B}_b\}_{b\in\F_q\cup\{\infty\}}$};
			\node [nonterminal, text width = 1.2cm, right of= init, xshift= 3em,yshift= -3em](H){$\bm{H}^{(q)}$};
			\node [nonterminal, text width = 1.2cm, below of = Net,yshift=-3em](HS){$\tilde{\bm{H}}^{(q)}$};
			\node [terminal, text width = 1.0cm, right of = B, xshift= 5em](OutD){$\bm{D}$};	
			\node [terminal, text width = 1.0cm, above of = OutD, yshift= 3em](Outx){$\bm{x}$};
			
			\draw [color=purple,thick,->](init) -- (GF);
			\draw [color=purple,thick,->](GF) -- node[pos=0.5,below,yshift=0em]{\eqref{matrixG}}(Gs);		
			\draw [color=purple,thick,->](Gs) -- node[pos=0.5,below,yshift=0em]{\eqref{latin}}(L);		
			\draw [color=purple,thick,->](L) -- node[pos=0.5,below,yshift=0em]{\eqref{q-th block}}node 		  [pos=0.5,above,yshift=0em]{\eqref{ks-net}}(Net);	
			\draw [color=purple,thick,->](B) -- node[pos=0.5,below,yshift=-0.1em]{\eqref{eq:dic}}(OutD);
			\draw [color=purple,thick,->](init) -- (H);
			\draw [color=purple,thick,->](H) -- node[pos=0.5,below,yshift=-0.1em]{\eqref{eq:hhs}}(HS);				
			
			\draw [color=purple,thick,->](Net)|- node[pos=0,below,xshift=1.2em,yshift=-2.2em]{\eqref{eq:mubss}}(B);				
			\draw [color=purple,thick,->](HS) |- (B);			
			
			\draw [color=purple,thick,->](Gs) |- node[pos=0.5,above,xshift=13em,yshift=0em]
			{\eqref{eq:x^i} \eqref{eq:x} }(Outx);				
			
		\end{tikzpicture}

		\caption{ Process flow diagram that explains the construction of dictionaries $\bm{D}$ and vectors $\x$.} \label{fig1}
		
	\end{center}
\end{figure*}

\subsection{Outline}

The rest of this paper is organized as follows. Section \ref{sec:GF}  defines
two kinds of matrices by using elements in Galois fields. One of them contains mutual orthogonal Latin Squares, the other is used for
theoretical analysis.
Section \ref{sec:hadamard} defines a class of real Hadamard matrices.
Section \ref{sec:mub} obtains mutual unbiased bases in square dimension with explicitly structures.
Section \ref{sec:main} proves main results and answers the open problem using the mutual unbiased bases constructed in section \ref{sec:mub}.
Section \ref{sec:exa} presents three examples to illustrate our constructions and theoretical proofs.
Section \ref{sec:con} gives conclusions and further remarks.

\begin{table}[tbhp]
	{
		\footnotesize
		\caption{Notation in Figure \ref{fig1}. } \label{tab:1}
		\begin{center}
			\renewcommand\arraystretch{1.5}
			\begin{tabular}{|c|c|}
				\hline
				Galois Field of order $q$ &  $\mathbb{F}_q$ \\ \hline
				Matrices of order $q$ over $\mathbb{F}_q$ &  $ \bm{G}^{(q)}$ \\ \hline
				Latin Squares over $\F_q$ & $\{\bm{L}^{r}\}_{r\in\F_q}$\\ \hline
				Mutually Unbiased Bases for $\mathbb{R}^{q^2}$ & $\{\bm{B}_b\}_{b\in\F_q\cup \infty}$\\ \hline
				Hadamard matrix of order $q$ & $\bm{H}^{(q)}$\\ \hline
				Permutated Hadamard matrix of order $q$ & $\tilde{\bm{H}}^{(q)}$\\ \hline
				Dictionary of size $(q^2, q^2(q+1))$.
				& $\bm{D}$\\ \hline
				Sparse vector in $\mathbb{R}^{q^2(q+1)}$ & $\x$\\ \hline
			\end{tabular}
		\end{center}
	}
\end{table}

\section{Families of matrices} \label{sec:GF}
We briefly recall the Galois Field and Latin Squares in discrete mathematics \cite{denes2015}.
Then we construct two families of matrices by using  operations of Galois field $\F_q$ and establish some properties.

\subsection{Galois fields}
Write  $\F_2=\{0,1\}$ for the prime filed of order $2$. Let $m$ be a positive integer. In the rest of this paper, we assume that  $q=2^m$.
The \textit{Galois Field} of order
$q$ is a finite field of characteristic $2$, denoted by
$
\mathbb{F}_q.
$
As a vector space over $\F_2$, $\F_q$ is $m$-dimensional, and so the elements of $\F_q$ have a one to one correspondence to ones of $\F_2^m$.

More precisely, for any $i$ in $\F_q$, there exist $\omega_1,\omega_2,\cdots,\omega_m$ in $\F_2$ such that $i$ can be represented as follows
\begin{equation*}
	i=\omega_1\omega_2\cdots\omega_m,
\end{equation*}
with respect to some basis of the vector space $\F_q$ over $\F_2$. In particular, $0\in\F_q$ can represented by
$$
\underbrace{00\cdots0}_m.
$$

%

Naturally, there is an $\F_q$-indexed square matrix over $\mathbb{F}_q$ arisen by the multiplication table of $\mathbb{F}_q$, denoted by $\bm{G}^{(q)}$, that is, the $(i,j)$-th entry is
\begin{equation}\label{matrixG}
	g^{(q)}_{i,j}=ij
\end{equation}
for any $i$ and $j$ in $\mathbb{F}_q$.
\begin{lemma} \label{lem:diagonal of G}
	For any given $q=2^m$, the diagonal of $ \bm{G}^{(q)} $ is a permutation of the elements of $ \mathbb{F}_q $.
\end{lemma}
\begin{proof}
It is an immediate consequence of the fact that the Frobenius map of $\F_q$ is a bijection. For completeness, we give a whole proof.

Write $Z$ for the set $\{a^2\mid a\in \mathbb{F}_q\}$. It suffices to show that the set $Z$ equals $\mathbb{F}_q$. Suppose $a^2=b^2$ where $a,b\in \mathbb{F}_q$. Since the characteristic of the Galois filed $\mathbb{F}_q$ is $2$, we have $b^2=-b^2$, $2ab=0$, and
	$$(a- b)^2=a^2+b^2=a^2-b^2=0.$$
	Hence, $a=b$. It implies there are $q$ elements in $Z$, and $Z=\mathbb{F}_q$.
\end{proof}

\subsection{Two  famlilies of matrices} \label{subset:L}

A \textit{Latin square} of order $ q $ is a square matrix of order $q$  with entries from a set of cardinality $ q $ such that each element occurs once in each row and each column.  Two Latin squares $ \bm{L} $ and $ \bm{L}' $  are said to be \textit{orthogonal}  if all the ordered pairs are different.
A collection of Latin squares of order $ q $, any pair of which is orthogonal, is called a set of \textit{mutually orthogonal Latin squares}. Any Galois field $\F_q$  generates $q-1$ different orthogonal Latin squares of $q$ symbols.

We define a family of $\F_q$-indexed square matrices $ \{\bmL^r\}_{r\in \mathbb{F}_q} $ over $\mathbb{F}_q$,  where the $(i,j)$-th entry of $\bmL^r$ is
\begin{equation}\label{latin}
	l^r_{i,j}=g^{(q)}_{i,r}+j,
\end{equation}
for any $i,j,r$ in $\mathbb{F}_q$.
\begin{rem}
	It is not hard to check that $\{\bm{L}^r\}_{r\in \F_q^*}$ is a family of $q-1$ different orthogonal Latin squares, where $\mathbb{F}^*_q=\F_q-\{0\}.$
\end{rem}

There is a useful property of $\{\bm{L}^r\}_{r\in \F_q}$.

\begin{pro}\label{pro_latin}
	Let $\{\bm{L}^r\}_{r\in \F_q}$ be the family of squares matrices defined in \eqref{latin}.
	\begin{enumerate}
		\item For any  distinct $j_1,j_2\in \F_q$ and any $i,r\in\F_q $,  $ l^r_{i,j_1}\neq l^r_{i,j_2}$.
		\item For any distinct $r,r'\in \F_q $ and any $j_1, j_2\in\F_q$, there exists a unique $ i\in\F_q $ such that
		$$ l^{r}_{i,j_1}= l^{r'}_{i,j_2}.$$
	\end{enumerate}
\end{pro}
\begin{proof}
	\begin{enumerate}
		\item Since $j_1\neq j_2$, we have
		$$
		g^{(q)}_{i,r}+ j_1\neq g^{(q)}_{i,r}+ j_2	
		$$
		for any $i,r\in\F_q$. It implies that $ l^r_{i,j_1}\neq l^r_{i,j_2}$ by \eqref{latin}.
		
		
		\item   Since $ r\neq r'$, it is easy to deduce that
		$$
		\mathbb{F}_q=\left\{(r- r') i\mid i\in\F_q\right\}.
		$$
		For any $j_1,j_2\in \F_q$, there exits a unique $ i\in \F_q $ such that $
		(r- r') i=j_2- j_1.$
		Then
		$$
		l_{i,j_1}^r=g^{(q)}_{i,r}+j_1= ri+ j_1= r'i+ j_2=g^{(q)}_{i,r'}+j_2=l^{r'}_{i,j_2}.
		$$\qed
	\end{enumerate}	
	The proof is completed.
\end{proof}

We choose one column from each of Latin squares $\{\bm{L}^r\mid r\in\F_q\}$ to define a new $\F_q$-indexed square matrix $\bm{A}^{(q)}$ over $\F_q$ such that the $j$-th column $\bm{a}^{(q)}_{j}$ of $\bm{A}^{(q)}$ satisfies
\begin{equation}\label{A_L}
	\bm{a}^{(q)}_{j}=\bm{l}^j_{j^2},\quad j\in\F_q,
\end{equation}
where $\bm{l}^j_{j^2}$ is the ${j^2}$-th column vector of $\bmL^j$ for any $j\in\mathbb{F}_q$.
More precisely, the $(i,j)$-th entry $a_{i,j}^{(q)}$ of $\bm{A}^{(q)}$ is
\begin{equation}\label{eq:L and A}
	a^{(q)}_{i,j}=ij+j^2=l_{i,j^2}^{j},\quad i,j\in\F_q.
\end{equation}

\begin{thm}\label{index}
	For any given $ q=2^m $, the $0$-th row of $ \bmA^{(q)} $ is a permutation of the elements of $ \mathbb{F}_q $, while $ a^{(q)}_{i,j_1}=a^{(q)}_{i,j_2} $ if and only if $ j_1+ j_2=i $ for any $ i,j_1, j_2\in\F_q$ and $j_1\neq j_2$.
\end{thm}
\begin{proof} 
	%
	Clearly, each entry of $0$-th row of $\bm A^{(q)}$ is $a^{(q)}_{0,j}=j^2=g^{(q)}_{j,j}$ for $j\in\F_q$. It is exactly the diagonal of $\bm G^{(q)}$, then the first result follows from Lemma \ref{lem:diagonal of G}.
	
	For any given $i\in\F_q$, assume $ j_1+ j_2=i$ where $j_1,j_2\in F_q$ and $j_1\neq j_2$. We have
	\begin{align*}
		a^{(q)}_{i,j_1}=i j_1+ j_1^2
		=\left((j_1+ j_2) j_1\right)+ j_1^2
		=j_2 j_1,
	\end{align*}
		and
			\begin{align*}
		a^{(q)}_{i,j_2}=i j_2+ j^2_2
		=\left((j_1+ j_2) j_2\right)+ j_2^2
		=j_1 j_2.
	\end{align*}
	Hence, $ a^{(q)}_{i,j_1}=a^{(q)}_{i,j_2} $.

	If $ a^{(q)}_{i,j_1}=a^{(q)}_{i,j_2}, $ where $j_1,j_2\in\F_q$ and $j_1\neq j_2$, that is,
	$$
	i j_1+ j_1^2= i j_2+ j^2_2.
	$$
	Adding $i j_2+ j_1^2$ to both sides of the equation above, one obtains that
	$$
	i(j_1+ j_2)=j_1^2+ j_2^2=(j_1+ j_2)^2.
	$$
	Then $i=j_1+ j_2$ holds, since $j_1+ j_2\neq 0$ in $\mathbb{F}_q$.
\end{proof}

\section{Hadamard Matrix}\label{sec:hadamard}
The construction of MUBs in \cite{Wocjan2004}  are intimately linked to  Hadamard matrices.
For the desired construction of sparse vectors in the null space, this section focuses on a special  family of real Hadamard matrices.
Recall that a  \textit{Hadamard matrix} $ \bmH $ is a square matrix whose entries are either $ +1 $ or $ -1 $ and whose columns are mutually orthogonal.
Let
$$
\bmH^{(2)}=\begin{pmatrix}
	1&1\\
	1&-1
\end{pmatrix}
$$
be a Hadamard matrix of order $2$. In a natural way, the Hadamard matrix $\bmH^{(2)}$ can be viewed as an $\F_2$-indexed matrix, that is, the entries of $\bmH^{(2)}$ are 
$
h^{(2)}_{0,0}=h^{(2)}_{0,1}=h^{(2)}_{1,0}=1$ and $h^{(2)}_{1,1}=-1.
$


For any given positive integer $ m $,  by
repeating used of the Hadamard matrix   $\bmH^{(2)}$, a Hadamard matrix of order $q=2^m$ can be obtained as follows \cite{Sylvester1867}	
$$
\bmH^{(q)}=\underbrace{\bmH^{(2)}\otimes\dots\otimes \bmH^{(2)}}_{m\ \text{times}}.
$$	
Clearly, $\bmH^{(q)}$ is an  $\F_2^m$-indexed matrix. To be precise, the $\left((\omega_1\omega_2\cdots\omega_m),(\nu_1\nu_2\cdots\nu_m)\right)$-th entry of $\bmH^{(q)}$ is
		$$
		h^{(2)}_{\omega_1,\nu_1}h^{(2)}_{\omega_2,\nu_2}\cdots h^{(2)}_{\omega_m,\nu_m},
		$$
for any 		$\omega_1$, $\omega_2$, $\cdots$, $\omega_m$, $\nu_1$, $\nu_2$, $\cdots$, $\nu_m\in \F_2$.  Such Hadamard matrices have many
applications in computer science and  quantum information \cite{Horadam2007}.

%
%

Let $\varphi$ be a bijective self-mapping of $\F_2^m$ such that $\varphi(0)=0$, and $\bmH_{\varphi}^{(q)}$ an $\F_2^m$-indexed matrix satisfying the $(\omega_1\cdots\omega_m)$-th row of $\bmH_{\varphi}^{(q)}$ is the $\varphi(\omega_1\cdots\omega_m)$-th row of $\bmH^{(q)}$ for any $\omega_1,\cdots,\omega_m\in\F_2$. It is clear that $\bmH^{(q)}_{\varphi}$ is also a Hadamard matrix.

To our end, we choose such a bijection $\sigma$ as follows
$$
\begin{array}{cclc}
\sigma:& \F_2^m&\to &\F_2^m\\
~& \omega_1\omega_2\cdots \omega_m&\mapsto &\omega'_1\omega'_2\cdots \omega'_m,
\end{array}
$$	
where if $k$ is the largest integer such that $\omega_k$ is nonzero,
 $\omega'_1\omega'_2\cdots \omega'_m$ satisfies
\begin{equation}\label{eq:hhs}
	\begin{cases}
		\omega_s'+\omega_s=1,\quad s=1,\dots,k-1,\\
		\omega_k'=\omega_k=1,\\
		\omega_s'=\omega_s=0,\quad s=k+1,\dots,m.
	\end{cases}	
\end{equation}
In the sequel, we denote the permuted Hadamard matrix ${\bmH}_{\sigma}^{(q)}$ by $\tilde{\bmH}^{(q)}$.

Since each element of $\F_q$ has a representation by one of $\F_2^m$ with respect to some basis of the vector space $\F_q$ over $\F_2$,  the Hadamard matrix $\bmH^{(q)}$ and  $\tilde{\bmH}^{(q)}$ are also $\F_q$-indexed matrices. We establish  useful properties of the permuted Hadamard matrix $\tilde{\bmH}^{(q)}$.


%


\begin{thm}\label{a2hadamard}
	For any given $ q=2^m $, the permuted Hadamard matrix $ \tilde{\bmH}^{(q)} $ satisfies the following conditions,
	\begin{enumerate}
		\item the entries in $0$-th row and $0$-th column of $ \tilde{\bmH}^{(q)} $ are $ 1 $;
		\item $ \tilde{h}^{(q)}_{i,j_1}=-\tilde{h}^{(q)}_{i,j_2} $ for $i\neq 0$, $j_1+ j_2=i$ and  $ i,j_1,j_2\in \F_q$,
	\end{enumerate}
	where $\tilde{h}^{(q)}_{i,j}$ is $(i,j)$-th entry of $\tilde{\bmH}^{(q)}$ for any $i,j\in \F_q$.
\end{thm}
\begin{proof}
	We fix a basis of the vector space $\F_q$ over $\F_2$ in this proof. Write  ${h}^{(q)}_{i,j}$ for $(i,j)$-th entry of ${\bmH}^{(q)}$ for any $i,j\in \F_q$.
	
	Since $\sigma(0)=0$, the $0$-th row of $\tilde{\bmH}^{(q)}$ is the $0$-th row of ${\bmH}^{(q)}$ whose entries are all $1$. For any $j\in\F_q$,  it is not hard to check that $\tilde{h}^{(q)}_{j,0}$ is a multiplication of elements of $0$-th column of $\bmH^{(2)}$, and so it equals $1$.
			
	Now we assume $i=\omega_1\omega_2\cdots \omega_m$ is a nonzero element of $\F_q$, where $\omega_1,\omega_2,\cdots ,\omega_m\in\F_2$. Let $i'=\sigma(\omega_1\omega_2\cdots \omega_m)=\omega'_1\omega'_2\cdots \omega'_m\in\F_q$ where  $\omega'_1,\omega'_2,\cdots ,\omega'_m\in\F_2$ satisfies \eqref{eq:hhs}. By the definition of the permuted Hadamard matrix, the $i$-th row of $\tilde{\bmH}^{(q)}$ is exactly the $i'$-th row of $\bmH^{(q)}$.  For any $j_1,j_2\in\F_q$ such that $j_1+j_2=i$,
	$
	\tilde{h}^{(q)}_{i,j_1}=-\tilde{h}^{(q)}_{i,j_2}
	$ is equivalent to  $
	h^{(q)}_{i',j_1}=-h^{(q)}_{i',j_2}.
	$

	Write $j_1=\nu_1\nu_2\cdots\nu_m$ and $j_2=\nu'_1\nu'_2\cdots\nu'_m$ where $\nu_1,\cdots\nu_m,\nu'_1,\cdots\nu'_m\in\F_2$. Then
	$$
	h_{i'j_1}^{(q)}=h^{(2)}_{\omega'_1,\nu_1}h^{(2)}_{\omega'_2,\nu_2}\cdots h^{(2)}_{\omega'_m,\nu_m},
	$$
	and
	$$
	h_{i'j_2}^{(q)}=h^{(2)}_{\omega'_1,\nu'_1}h^{(2)}_{\omega'_2,\nu'_2}\cdots h^{(2)}_{\omega'_m,\nu'_m}
	$$
	The condition  is
	$$
	j_1+j_2=\nu_1\nu_2\cdots\nu_m+\nu'_1\nu'_2\cdots\nu'_m=\omega_1\omega_2\cdots \omega_m.
	$$ Notice that the addition of $\nu_1\nu_2\cdots\nu_m$ and $\nu'_1\nu'_2\cdots\nu'_m$ in the vector space $\F_q$ over $\F_2$ is a kind of the binary XOR operation. Let $k$ be the largest integer such that $\omega_k$ is nonzero. There are three cases.
	
	

	\begin{itemize}
		\item If $ s=1,\dots,k-1 $, then $ \omega_s'+\omega_s=1 $, and either $ \omega_s=0$, $\omega'_s=1 $ or
		$ \omega_s=1, \omega'_s=0 $.
		\begin{enumerate}
			\item[(1)] If $\omega_s=0$, $\omega'_s=1$, then $\nu_s=\nu'_s$. Hence,
			$$
			h^{(2)}_{\omega'_s,\nu_s}=h^{(2)}_{\omega'_s,\nu'_s}.
			$$
			
			\item[(2)]If $ \omega_s=1, \omega'_s=0 $, one obtains that			$$
			h^{(2)}_{\omega'_s,\nu_s}=1=h^{(2)}_{\omega'_s,\nu'_s}.
			$$
		\end{enumerate}
		\item
		If $ s=k $, then $ \omega_s'=\omega_s=1 $, and either  $\nu_s=1,\nu'_s=0$ or $\nu_s=0,\nu'_s=1$. We have
		\[h^{(2)}_{\omega'_s,\nu_s}=-1,\ h^{(2)}_{\omega'_s,\nu'_s}=1\]
		or
		\[h^{(2)}_{\omega'_s,\nu_s}=1,\ h^{(2)}_{\omega'_s,\nu'_s}=-1.\]
		It implies that
		$$
		h^{(2)}_{\omega'_s,\nu_s}=-h^{(2)}_{\omega'_s,\nu'_s}.
		$$
		
		\item
		If $ s=k+1,\dots,m $, then $ \omega_s'=\omega_s=0 $, and 
		$$
		h^{(2)}_{\omega'_s,\nu_s}=h^{(2)}_{\omega'_s,\nu'_s}=1.
		$$
		
	\end{itemize}
	In conclusion, $
	h^{(q)}_{i',j_1}=-h^{(q)}_{i',j_2}.
	$
\end{proof}

\section{MUBs in Square Dimensions}\label{sec:mub}

The concept of MUBs plays an important role in quantum information theory \cite{Thomas,kibler2018}.
MUBs are  uniform tight frames which have  been well studied in computational harmonic analysis \cite{2007Osustik,2003Grassmannian}.
There exist numerous ways of constructing sets of MUBs, see e.g. \cite{kibler2018,Thomas,Wocjan2004} and references therein.
This section  recalls the  method introduced in \cite{Wocjan2004}.
Following the line in \cite{Wocjan2004},  we  construct $ (q+1,q) $-net from
mutually orthogonal Latin squares obtained in subsection \ref{subset:L}. Then we define
a family of MUBs  using  $(k, q)$-nets and permuted Hadamard matrices obtained in section \ref{sec:hadamard}.

In the sequel, write $\bm{1}_q$ for the $\F_q$-indexed column vector over $\mathbb{R}$ with all entries that are equal to one, and
$ \bm{e}_{i} $ for the $i$-th column vector of $\F_q$-indexed identity matrix $ \bmI $ over $\mathbb{R}$ for any $i$ in $\F_q$.

\subsection{$ (q+1,q) $-net} \label{sect:net}

A column vector $ \m:=(m_1,\dots,m_d)^T $ of size $ d $ is called to be an \textit{incidence vector} if its entries take only the values
$0$ and $1$. Nets are  collections of incidence vectors satisfying special properties. The definition of $(q+1,q)$-net is from the design theory \cite{colbourn1996}.

\begin{defi}[$ (k,q) $-net]\cite[Definition 1]{Wocjan2004}\label{defi_net}
	Let $K$ be an indexed set with $k$ elements, and $\{\m_{b,j}\mid j\in \mathbb{F}_q\}$ be a set of incidence vectors of size $ d=q^2 $ for any $b\in K$. The collection of incidence vectors
	$$
	\left\{\{\m_{b,j}\}_{j\in\mathbb{F}_q}\mid b\in K\right\}
	$$
	is called a $(k,q)$-net,  if the following conditions hold
	\begin{enumerate}
		\item $ \langle\m_{b,i},\m_{b,j} \rangle =0 $, for any  distinct $i,j\in \F_q$ and $b\in K$;
		\item $ \langle \m_{b,i},\m_{c,j } \rangle =1 $, for any distinct $ b,c\in K$ and $i,j\in \mathbb{F}_q $.
	\end{enumerate}
\end{defi}

The relationship between mutually orthogonal Latin squares and nets has been discussed in \cite{Wocjan2004}. Let $K=\F_q\cup \{\infty\}$.
Using the square matrices $\{ \bmL^r \}_{r\in \mathbb{F}_q}$ over $\F_q$ obtained in \eqref{latin}, we construct a collection of incidence vectors to be a $(q+1,q)$-net as follows,
\begin{enumerate}
	\item
	If $ b\in \F_q$, $\m_{b,j}$ is an $\F_q^2$-indexed vector consisting of $q$-blocks $\{\bm{e}_{l^b_{u,j}}\mid u\in\F_q\}$, that is,
	\begin{equation}\label{ks-net}
		\m_{b,j}= \left(\bm{e}_{l^b_{u,j}}\right)_{u\in\F_q},\qquad j\in\F_q.
	\end{equation}
	\item
	If $b=\infty$,  $\m_{b,j}$ is an $\F_q^2$-indexed vector satisfying
	\begin{equation}\label{q-th block}
		\m_{\infty,j}=\bm{e}_j\otimes \bm{1}_q,\quad j\in \F_q.
	\end{equation}
	Similar to the case of $b\in\F_q$, $\m_{\infty,j}$ also consists of $q$-blocks, where $(m_{\infty,j})_{j}= \bm{1}_q$ and $(m_{\infty,j})_{u}=\bm{0}_q$ if $u\neq j$, for any $j\in\F_q$.
\end{enumerate}

\begin{thm}\label{thm:net} The collection of incidence vectors
	$$
	\left\{\{\m_{b,j}\}_{j\in\F_q}\mid b\in \F_q\cup\{\infty\}\right\}
	$$
	obtained in \eqref{ks-net} and \eqref{q-th block} is a $(q+1,q)$-net.
\end{thm}
\begin{proof}
	
	Let $b$ be a given element of $\F_q $. For any distinct $j_1, j_2\in\F_q $, then $l^b_{u,j_1}\neq l^b_{u,j_2}$ by Proposition \ref{pro_latin}. Hence,
	$$
	\left\langle\m_{b,j_1},\m_{b,j_2} \right\rangle=\sum_{u\in\F_q}
	\left\langle\bm{e}_{l^b_{u,j_1}},\bm{e}_{l^b_{u,j_2}} \right\rangle
	=0.
	$$
	For any distinct $b$, $b'\in \F_q $ and any $j_1$, $j_2\in\F_q$,  there exists a unique $ i\in\F_q $ such that
	$ l^{b}_{i,j_1}= l^{b'}_{i,j_2}$ by Proposition \ref{pro_latin}. One obtains
\begin{align*}
\langle\m_{b,j_1},\m_{b',j_2} \rangle
=\sum_{u\neq i}
\left\langle\bm{e}_{l^b_{u,j_1}},\bm{e}_{l^{b'}_{u,j_2}} \right\rangle
+
\left\langle\bm{e}_{l^b_{i,j_1}},\bm{e}_{l^{b'}_{i,j_2}} \right\rangle
=1.		
\end{align*}
	It's easy to verify that for any distinct $j_1$, $j_2\in\F_q$,
	$$
	\langle\m_{\infty,j_1},\m_{\infty,j_2} \rangle=0.
	$$
	For any $ b$, $j_1$, $j_2\in\F_q $,  we have
$$
\langle\m_{\infty,j_1},\m_{b,j_2} \rangle
=\sum_{u\in \F_q}
\left\langle
\left(\m_{\infty,j_1}\right)_u, \bm{e}_{l^b_{u,j_2}}\right\rangle
=\left\langle\bm{1}_q,\bm{e}_{l^{b}_{j_1,j_2}} \right\rangle=1.
$$
The proof is completed.
\end{proof}

\subsection{MUBs}
Let
$$\bm{B}_0=\begin{pmatrix}
	\phi_1 & \phi_2 & \cdots & \phi_d
\end{pmatrix}
$$
and
$$
\bm{B}_1=\begin{pmatrix}
	\psi_1 & \psi_2 & \cdots & \psi_d
\end{pmatrix}
$$
denote orthonormal bases in the $d$-dimensional space. Then they are said to be \textit{mutual unbiased} if and only if
\begin{equation}\label{eq:mutual}
	|\langle \phi_i, \psi_j  \rangle | = \frac{1}{\sqrt{d}}
\end{equation}
for all $i$, $j$.
The quantity \eqref{eq:mutual} is called the mutual coherence in compressed sensing \cite{donoho2001}.
A set $ \{\bm{B}_0, \ldots, \bm{B}_s \}$ of orthonormal bases  is said to be a set of \textit{mutually unbiased bases} (MUB) if and only if every pair of bases in the set is mutually unbiased.

Let
$$
\left\{
\{\m_{b,j}\}_{j\in\F_q}\mid b\in \F_q\cup\{\infty\}
\right\}
$$
be the $(q+1,q)$-net obtained in \eqref{ks-net} and \eqref{q-th block}.
%
%
For any $\F_q$-indexed column vector $ \bm{h} $ over $ \mathbb{R}$, the embedding of $ \bm{h} $ into $ \mathbb{R}^{q^2} $ controlled by $ \m_{b,j} $ is an $\F_q^2$-indexed vector, denoted by $ \bm{h}\uparrow\m_{b,j} $, satisfying
\begin{align*}
	\bm{h}\uparrow\m_{b,j}&:=
	\left(h_u \bm{e}_{l^b_{u,j}}
	\right)_{u\in \F_q},
	\quad b,j\in \F_q,\\
	\bm{h}\uparrow\m_{\infty,j}&:=\bm{e}_j\otimes \bm{h},\quad j\in \F_q,
\end{align*}
where $ h_u $ is the $u$-th entry of $ \bm{h} $.
Using the permuted Hadamard matrix $\tilde{\bm{H}}^{(q)}$, we construct a collection of $\F^2_q\times \F_q$-indexed matrices
$$
\left\{(\bmB_b)_u\mid b\in\F_q\cup\{\infty\},u\in\F_q\right\}
$$
such that the $v$-th column of $(\bmB_b)_u$ is $\tilde{\bm{h}}_v^{(q)}\uparrow\bm{m}_{b,u}$, or equivalently
$$
\left(\bmB_b\right)_u=
\left(\tilde{\bm{h}}^{(q)}_v\uparrow\bm{m}_{b,u}\right)_{ v\in \F_q},
$$
where $ \tilde{\bm{h}}^{(q)}_v $ is the $v $-th column of $ \tilde{\bm{H}}^{(q)} $ for any $v\in\F_q$.
Then there is a collection of $\F_q^2$-indexed square matrices  $\{\bm{B}_b\mid b\in\F_q\cup\{\infty\}\}$ over $\mathbb{R}$, where $\bm{B}_b$ consists of $q$-blocks $\{(\bm{B}_b)_u\mid u\in \F_q\}$, that is
\begin{equation}\label{eq:mubss}
	\bm{B}_b=
	\frac{1}{\sqrt{q}}\left((\bm{B}_b)_u\right)
	_{u\in\F_q}.
\end{equation}


The following result follows from Theorem \ref{thm:net} and \cite[Theorem 3]{Wocjan2004}.
\begin{thm}\label{thm:mubs}
	The orthonormal  bases $ \{\bmB_b\mid b\in \F_q\cup\{\infty\}\} $ constructed in \eqref{eq:mubss}
	are $ q+1 $ mutually unbiased bases for $\mathbb{R}^{q^2}$.
	%
\end{thm}
\section{Proofs of Main results}\label{sec:main}

This section proves main results in this paper.

\subsection{Proof of Theorem \ref{thm:L0}}

Using the orthonormal  bases
$
\left\{\bmB_b\mid b\in\F_q\cup \{\infty\}\right\}
$
obtained in \eqref{eq:mubss}, we construct  a dictionary by
\begin{equation}\label{eq:dic}
	\bm{D}=(\bmB_b)_{b\in\F_q\cup \{\infty\}}
	\in\mathbb{R}^{q^2 \times q^2(q+1)}.
\end{equation}
The dictionary $\bm{D}$ constructed in \eqref{eq:dic} consists of $q+1$ block  matrices indexed by $\F_q\cup \{\infty\}$.
It follows from Theorem \ref{thm:mubs} that
$$
\mu(\bm{D})=1/q.
$$
We define a collection of one sparse vectors as follows
\begin{equation}\label{eq:x^i}
	\begin{cases}
		\x^{b}=\bm{e}_{b^2}\otimes\bm{e}_b\in \mathbb{R}^{q^2},\quad b\in\F_q,\\
		\x^\infty=-\bm{e}_0\otimes\bm{e}_0\in \mathbb{R}^{q^2}.
	\end{cases}
\end{equation}
Then we obtain a $q+1$ sparse column vector
\begin{equation} \label{eq:x}
	\x =
	\left(\x^b\right)_{b\in \F_q\cup\{\infty\}}\in \mathbb{R}^{q^2(q+1)}.
\end{equation}
\begin{thm}\label{main thm}
	The $q+1$ sparse  vector  $\x$  constructed in \eqref{eq:x^i} is in the null space of $\bm{D}$ in \eqref{eq:dic}, i.e.,
	$$
	\sqrt{q} \bm{D} \x
	=\sum_{b\in \F_q\cup\{\infty\}} \sqrt{q}\bm{B}_b\x^b
	= \bm{0}_{q^2}.
	$$
\end{thm}
\begin{proof}
	The proof is divided into two cases.
	\begin{enumerate}

		\item  For $ b\in\F_q $, we have
		\begin{align*}
			\sqrt{q}\bmB_{b}\x^b
			=((\bm{B}_b)_{b^2})\bm{e}_b=
			\left(\tilde{\bm{h}}^{(q)}_v\uparrow \bm{m}_{b,b^2}\right)_{v\in\F_q}\bm{e}_b
			=\tilde{\bm{h}}^{(q)}_b\uparrow\m_{b,b^2}\\
			=
			\left(\tilde{h}^{(q)}_{i,b}\bm{e}_{l^b_{i,b^2}}\right)_{i\in\F_q}
			=
			\left(\tilde{h}^{(q)}_{i,b}\bm{e}_{a^{(q)}_{i,b}}\right)_{i\in\F_q},
		\end{align*}	
		where the last equality follows from \eqref{eq:L and A}.
		For any $i\in \mathbb{F}_q$, there are two subcases
		\begin{enumerate}
			\item If $i=0$,
			$$
			\sum_{b\in\F_q}\tilde{h}^{(q)}_{0,b}\bm{e}_{a^{(q)}_{0,b}}=\sum_{b\in\F_q}\bm{e}_{a^{(q)}_{0,b}}=\bm{1}_q,
			$$
			where the first equality follows from Theorem \ref{a2hadamard} and the second one follows from Theorem \ref{index}.
			\item If $i\neq 0$,
			\begin{align*}
			\sum_{b\in\F_q}\tilde{h}^{(q)}_{i,b}\bm{e}_{a^{(q)}_{i,b}}
			=\sum_{b\in\F_q}
			\tilde{h}^{(q)}_{i,i-b}\bm{e}_{a^{(q)}_{i,i-b}}
			=\sum_{b\in\F_q}
			-\tilde{h}^{(q)}_{i,b}\bm{e}_{a^{(q)}_{i,b}},
			\end{align*}
			where the second equality follows from Theorem \ref{a2hadamard} and Theorem \ref{index}. So
			$$
			\sum_{b\in\F_q}\tilde{h}^{(q)}_{i,b}\bm{e}_{a^{(q)}_{i,b}}=\bm{0}_{q}.
			$$
		\end{enumerate}
		Hence,
		\begin{align*}
			\sum_{b\in\F_q}\sqrt{q}\bmB_{b}\x^b
			=\sum_{b\in\F_q}
			\left(\tilde{h}^{(q)}_{i,b}\bm{e}_{a^{(q)}_{i,b}}\right)_{i\in\F_q} 
			=\bm{e}_0\otimes \bm{1}_q.
		\end{align*}
		
		\item
		For $ b=\infty $, one obtains
		\begin{align*}
			\sqrt{q}\bm{B}_\infty\x^\infty
			=&-((\bm{B}_\infty)_{0})\bm{e}_0 \\
			=&-\left(\tilde{\bm{h}}^{(q)}_v\uparrow \bm{m}_{\infty,0}\right)_{v\in\F_q}\bm{e}_0\\
			=&-\tilde{\bm{h}}^{(q)}_0\uparrow \bm{m}_{\infty,0}\\
			=&-\bm{1}_q\uparrow \bm{m}_{\infty,0}\\
			=&-\bm{e}_0\otimes \bm{1}_q.
		\end{align*}

	\end{enumerate}

	Combining two cases, we have
\begin{align*}
\sqrt{q}\bm{D}\x=\sum_{b\in \F_q\cup\{\infty\}} \sqrt{q}\bm{B}_b\x^b
=\left(\bm{e}_0\otimes\bm{1}_q-\bm{e}_0\otimes\bm{1}_q\right)
=\bm{0}_{q^2}.	
\end{align*}
\end{proof}

Theorem \ref{thm:L0} follows from Theorem \ref{main thm} directly.

\subsection{Proof of Theorem \ref{thm:L02}}
In this subsection, we focus on the field extension $\F_q\subseteq \F_{q^2}$. Then we have a quotient group $\F_{q^2}/\F_q$. For any $i\in\F_{q^2}$,  write $[i]$ for the coset of $i$ in $\F_{q^2}/\F_q$.

Obviously, $\F_q$ and $\F_{q^2}/\F_q$ are both $1$-dimensional vector spaces over $\F_q$, and they are isomorphic to each other as vector spaces over $\F_q$. Write $\xi$ for an isomorphism from $\F_q$ to $\F_{q^2}/\F_q$.

\begin{lemma}\label{lem: change of b}
Let $i$ be an element of $\F_{q^2}$ and $b_*$ the element in $\F_q$ such that $\xi(b_*)=[i]$. Then
\begin{align*}
\quad\left\{ l^b_{i,j}\mid [j]=\xi(b^2),j\in\F_{q^2}\right\}
=\left\{l^{b_*-b}_{i,j}\mid [j]=\xi((b_*-b)^2),j\in\F_{q^2}\right\},
\end{align*}
for any $b\in\F_q$.
\end{lemma}

\begin{proof}
Write $S_b=\left\{ l^b_{i,j}\mid [j]=\xi(b^2),j\in\F_{q^2}\right\}$ for any $b\in\F_q$.
For any given $b\in\F_q$ and any $j\in \F_{q^2}$ such that $[j]=\xi(b^2)$, we have
\begin{align*}
[j+b_*i]=[j]+[b_*i]=\xi(b^2)+b_*[i]=\xi(b^2)+b_*\xi(b_*)=\xi(b^2)+\xi(b_*^2)=\xi((b_*-b)^2).
\end{align*}
in $\F_{q^2}/\F_q$, and
\[l^b_{i,j}=bi+j=(b_*-b)i+j+b_*i=l^{b_*-b}_{i,j+b_*i}.\]
Hence $S_b\subseteq S_{b_*-b}$. One also obtains that $S_{b_*-b}\subseteq S_{b_*-(b_*-b)}=S_b$. The result follows.
\end{proof}

Recall that each element of $\F_q$ (resp. $\F_{q^2}$) can be represented by one of $\F_2^m$ (resp. $\F_2^{2m}$) with respect to some basis of the vector space $\F_q$ (resp. $\F_{q^2}$) over $\F_2$. Note that $\F_q$ is a subfield of $\F_{q^2}$. In the sequel, we fix a basis of the vector space $\F_{q^2}$ over $\F_2$ such that any element of $\F_q$ in $\F_{q^2}$ has the form
$$
\omega_1\omega_2\cdots\omega_m\underbrace{00\cdots0}_m,
$$
where $\omega_1,\omega_2,\cdots,\omega_m\in\F_2$. Then one obtains that
$$
\F_{q^2}/\F_q=\{[00\cdots0\omega_1\omega_2\cdots\omega_m]\mid \omega_1,\omega_2,\cdots,\omega_m\in\F_2\}.
$$
We define a map 
$$
\begin{array}{cclc}
\iota: & \F_q &\to &\F_{q^2}\\
~& b &\mapsto & 0\cdots0\,\omega_1\cdots\omega_m,
\end{array}
$$
where  $[0\cdots0\omega_1
\cdots\omega_m]=\xi(b)$.
%

\begin{lemma}\label{lem: permuted matrix of q^2}
Let $b$ be any given element in $\F_q\subseteq \F_{q^2}$ and $i$ an element in $\F_{q^2}$.
\begin{enumerate}
\item If $i\in\F_q$, then $\tilde{h}^{(q^2)}_{i,\iota(b)}=1$.
\item If $i\notin \F_q$, then $\tilde{h}^{(q^2)}_{i,\iota(b)}+\tilde{h}^{(q^2)}_{i,\iota(b_*-b)}=0$ where $b_*$ is the unique nonzero element in $\F_q$ such that $\xi(b_*)=[i]$.
\end{enumerate}
\end{lemma}
\begin{proof}
\begin{enumerate}
\item Since $i\in \F_q$, we have 
$$
i=\omega_1\omega_2\cdots\omega_m00\cdots0,
$$
for some $\omega_1,\omega_2,\cdots,\omega_m\in\F_2$. Let $i'=\omega'_1\omega'_2\cdots\omega'_m00\cdots0$ be the element in $\F_{q^2}$ satisfying \eqref{eq:hhs}. Then $\tilde{h}_{i,\iota(b)}^{(q^2)}$ is the element $h_{i',\iota(b)}^{(q^2)}$ in the Hadarmard matrix $\bmH^{(q^2)}$. Write
$$
\iota(b)=00\cdots0\nu_1\nu_2\cdots\nu_m,
$$
for some $\nu_1,\nu_2,\cdots,\nu_m\in\F_2$. Then
$$
\tilde{h}^{(q^2)}_{i,\iota(b)}={h}^{(q^2)}_{i',\iota(b)}=h^{(2)}_{\omega'_10}\cdots
h^{(2)}_{\omega'_m0}h^{(2)}_{0\nu_1}\cdots h^{(2)}_{0\nu_m}=1.
$$

\item Since $i\notin \F_q$, we have
$$
i=\omega_1\omega_2\cdots\omega_m\omega_{m+1}\omega_{m+2}\cdots\omega_{2m},
$$
for some
$\omega_1$, $\omega_2$, $\cdots$, $\omega_{2m}\in\F_2$
and  $\omega_{m+1}$, $\omega_{m+2}$, $\cdots$, $\omega_{2m}$ are not all zero.  Let $i'$ be the element in $\F_{q^2}$ such that the $i'$-th row of the Hadamard matrix $\bmH^{(q^2)}$ is the $i$-th row of the permuted Hadamard matrix $\tilde{\bmH}^{(q^2)}$, denoted by $i'=\omega'_1\omega'_2\cdots\omega'_m\omega'_{m+1}\omega'_{m+2}\cdots\omega'_{2m}$. Write
$$
\iota(b)=00\cdots0\nu_1\nu_2\cdots\nu_m,
$$
and
$$
\iota(b_*-b)=00\cdots0\nu'_1\nu'_2\cdots\nu'_m,
$$
for some
$
\nu_1, \nu_2,  \cdots, \nu_m, \nu_1', \nu_2', \cdots, \nu_m'\in\F_2.
$
One obtains that
\begin{align*}
[00\cdots0\nu_1\nu_2\cdots\nu_m]
+[00\cdots0\nu'_1\nu'_2\cdots\nu'_m]
=[i]=[00\cdots0\omega_{m+1}\omega_{m+2}\cdots\omega_{2m}].
\end{align*}

We denote
\begin{align*}
j_1&=\nu_1\nu_2\cdots\nu_m, \\
j_2&=\nu'_1\nu'_2\cdots\nu'_m, \\
k&=\omega_{m+1}\omega_{m+2}\cdots\omega_{2m},  \\
k'&=\omega'_{m+1}\omega'_{m+2}\cdots\omega'_{2m},
\end{align*}
in $\F_q$. It is not hard to check that $k$ and $k'$ satisfy the condition \eqref{eq:hhs} and $j_1+j_2=k$.
%
%
By Theorem \ref{a2hadamard},
\begin{align*}
&\tilde{h}^{(q^2)}_{i,\iota(b)}+\tilde{h}^{(q^2)}_{i,\iota(b_*-b)}\\
=&{h}^{(q^2)}_{i',\iota(b)}+{h}^{(q^2)}_{i',\iota(b_*-b)} \\
=&h^{(2)}_{\omega'_1,0}\cdots h^{(2)}_{\omega'_m,0}h^{(2)}_{\omega'_{m+1},\nu_1}\cdots h^{(2)}_{\omega'_{2m},\nu_m}+h^{(2)}_{\omega'_1,0}\cdots h^{(2)}_{\omega'_m,0}h^{(2)}_{\omega'_{m+1},\nu'_1}\cdots h^{(2)}_{\omega'_{2m},\nu'_m}\\
=&h^{(2)}_{\omega'_{m+1},\nu_1}\cdots h^{(2)}_{\omega'_{2m},\nu_m}
 +h^{(2)}_{\omega'_{m+1},\nu'_1}\cdots h^{(2)}_{\omega'_{2m},\nu'_m}\\
=&{h}^{(q)}_{k',j_1}+{h}^{(q)}_{k',j_2}\\
=&\tilde{h}^{(q)}_{k,j_1}+\tilde{h}^{(q)}_{k,j_2}\\
=&0. 
\end{align*}
\end{enumerate}
\end{proof}
%
%
%
%
%
%
%
%

Using the orthonormal  bases
$$
\left\{\bmB_c\mid c\in\F_{q^2}\cup \{\infty\}\right\}
$$
obtained in \eqref{eq:mubss}, we construct  a dictionary by
\begin{equation}\label{eq:dic22}
	\bm{D}=(\bmB_b)_{b\in\F_q\cup \{\infty\}}
	\in\mathbb{R}^{q^4 \times q^4(q+1)}.
\end{equation}
The dictionary $\bm{D}$ constructed in \eqref{eq:dic22} consists of $q+1$ block  matrices indexed by $\F_q\cup \{\infty\}$.
It follows from Theorem \ref{thm:mubs} that
$$
\mu(\bm{D})=1/q^2.
$$
Write
\begin{equation}\label{eq:y^i}
	\left\{
	\begin{array}{ll}
		\y^{b}=\sum_{[j]=\xi(b^2)}\bm{e}_{j}\otimes\bm{e}_{\iota(b)}\in \mathbb{R}^{q^4}, & b\in\F_q\\
		\y^\infty=-\sum_{b\in\F_q}\bm{e}_{b}\otimes\bm{e}_{0}\in \mathbb{R}^{q^4}.
	\end{array}
	\right.
\end{equation}
Then we have a
$q^2+q$ sparse column vector
\begin{equation} \label{eq:y}
	\y=\left(\y^b\right)_{b\in \F_{q}\cup\{\infty\}}\in \mathbb{R}^{q^4(q+1)}.
\end{equation}
which consists of
{$q+1$}
blocks indexed by $\F_{q}\cup \{\infty\}$.

\begin{thm}\label{main thm2}
The $q^2+q$ sparse  vector  $\y$  constructed in \eqref{eq:y^i} is in the null space of $\bm{D}$ defined in \eqref{eq:dic22}, i.e.,
$$
q \bm{D} \y
=\sum_{b\in \F_{q}\cup\{\infty\}}q\bm{B}_b\y^b
= \bm{0}_{q^4}.
$$
\end{thm}
\begin{proof}
We consider two cases.
\begin{enumerate}

	
	\item  For $ b\in\F_q $, we have
\begin{align*}
	q\bmB_{b}\y^{b}
	=\sum_{[j]=\xi(b^2),}((\bm{B}_{b})_j)\bm{e}_{\iota(b)}
	=\sum_{[j]=\xi(b^2)}\tilde{\bm{h}}^{(q^2)}_{\iota(b)}\uparrow\m_{b,j}
	=\sum_{[j]=\xi(b^2)}
	\left(\tilde{h}^{(q^2)}_{i,\iota(b)}\bm{e}_{l^{b}_{i,j}}\right)_{i\in\F_{q^2}}.
	\end{align*}	
For any $i\in\F_{q^2}$, there are two subcases.
\begin{enumerate}
\item If $i\in\F_q$,
\begin{align*}
\sum_{b\in\F_q}\sum_{[j]=\xi(b^2)}\tilde{h}^{(q^2)}_{i,\iota(b)}\bm{e}_{bi+j}=\sum_{b\in\F_q}\sum_{[j]=\xi(b^2)}\bm{e}_{j}
&=\bm{1}_{q^2},
\end{align*}
where the first equality follows from Lemma \ref{lem: permuted matrix of q^2} and $[j+bi]=[j]$ in $\F_{q^2}/\F_q$ for any $j\in\F_{q^2}$ since $bi\in \F_q$, and the second one follows from the fact cosets determining a partition of $\F_{q^2}$ and Lemma \ref{lem:diagonal of G}.

\item If $i\notin\F_q$, then $[i]=\xi(b_*)$ for some nonzero ${b_*}\in\F_q.$ One obtains that
\begin{align*}
\sum_{b\in\F_q}\sum_{[j]=\xi(b^2)}\tilde{h}^{(q^2)}_{i,\iota(b)}\bm{e}_{l^{b}_{i,j}} 
=
\sum_{b\in\F_q}\sum_{[j]=\xi((b_*-b)^2)}\tilde{h}^{(q^2)}_{i,\iota(b_*-b)}\bm{e}_{l^{b_*-b}_{i,j}}
=
\sum_{b\in\F_q}\sum_{[j]=\xi(b^2)}-\tilde{h}^{(q^2)}_{i,\iota(b)}\bm{e}_{l^{b}_{i,j}},
%
\end{align*}
where the second equality follows from Lemma \ref{lem: change of b} and Lemma \ref{lem: permuted matrix of q^2}. So
$$
\sum_{b\in\F_q}\sum_{[j]=\xi(b^2)}\tilde{h}^{(q^2)}_{i,\iota(b)}\bm{e}_{l^{b}_{i,j}}
=\bm{0}_{q^2}.
$$

\end{enumerate}

Hence,
\begin{align*}
\sum_{b\in\F_{q}}q\bmB_{b}\y^b
=\sum_{b\in\F_q}\sum_{[j]= \xi(b)}
\left(\tilde{h}_{i,\iota(b)}^{(q^2)}\bm{e}_{l^{b}_{i,j}}\right)_{i\in\F_{q^2}}
=\sum_{b\in\F_q}\bm{e}_b\otimes \bm{1}_{q^2}.
\end{align*}

	\item
For $ b=\infty $, one obtains
\begin{align*}
	q\bm{B}_\infty\y^\infty
	&=-\sum_{b\in\F_q}((\bm{B}_\infty)_{b})\bm{e}_{0}\\
	&=-\sum_{b\in\F_q}\left(\tilde{\bm{h}}^{(q^2)}_v\uparrow \bm{m}_{\infty,b}\right)_{v\in\F_{q^2}}\bm{e}_{0}\\
	&=-\sum_{b\in\F_q}\tilde{\bm{h}}_{0}^{(q^2)}\uparrow \bm{m}_{\infty,b}\\
	&=-\sum_{b\in \F_q}\bm{e}_{b}\otimes \tilde{\bm{h}}^{(q^2)}_{0}\\
	&=-\sum_{b\in \F_q}\bm{e}_{b}\otimes \bm{1}_{q^2}.
\end{align*}

\end{enumerate}
The proof is completed.
\end{proof}

Theorem \ref{thm:L02} follows from Theorem \ref{main thm2} directly.

\section{Examples}\label{sec:exa}

This section provides three examples to illustrate the  process  showed in Figure \ref{fig1}.
In order to more easily visualize the dictionaries and sparse vectors in the null space, we map the $+1$ elements as
red squares, the $-1$ elements as blue squares and $0$ elements as gray squares. Figure \ref{fig:d1}, Figure \ref{fig:d2} 
and Figure  \ref{fig:d333} show  the dictionaries and vectors in Example \ref{exa:2}, Example \ref{exa:3} 
and Example \ref{exa:4} respectively.

\begin{Exa}\label{exa:2}

	\begin{table}[H]
		\centering
		\begin{tabular}{c|cc}
			$\cdot$&0&1 \\
			\hline
			0&0&0\\
			1&0&1
		\end{tabular}	
		\caption{Multiplication table  for $\F_2$}\label{Tab1M}
	\end{table}	
	
	\begin{table}[H]
		\centering
		\begin{tabular}{c|cc}
			$+$&0&1 \\
			\hline
			0&0&1\\
			1&1&0
		\end{tabular}
		
		\caption{Addition table for $\F_2$}\label{Tab1A}
	\end{table}	
	
Let $q=2$. It follows that
	$$
	\F_2=\{0,1\}.
	$$
The multiplication table and addition table for $\F_2$ are shown in Table \ref{Tab1M} and Table \ref{Tab1A}.
The matrix $\bm{G}^{(2)}$ defined in \eqref{matrixG} is
	\begin{equation}\label{eq:g2}
		\bm{G}^{(2)} =
		\begin{pmatrix}
			0&0\\
			0&1
		\end{pmatrix}.
	\end{equation}
	Then the Latin squares (see \eqref{latin}) are
	$$
	\bmL^0=
	\begin{pmatrix}
		0 & 1\\
		0 & 1
	\end{pmatrix},\quad
	\bmL^1=
	\begin{pmatrix}
		0 & 1\\
		1 & 0
	\end{pmatrix}.
	$$The matrix $\bm{A}^{(2)}$ defined in \eqref{A_L} (or \eqref{eq:L and A}) is
	$$
	\bmA^{(2)}=	
	\begin{pmatrix}
		0&1\\
		0&0
	\end{pmatrix}.
	$$
	Six incidence vectors of $ (3,2) $-net constructed in \eqref{ks-net} and \eqref{q-th block} are as follows:
$$
\m_{0,0}=\begin{pmatrix}
			\bm{e}_{0}^T&
			\bm{e}_{0}^T
		\end{pmatrix}^T
		=\begin{pmatrix}
			1&0&1&0
		\end{pmatrix}^T,
$$
$$
	\m_{0,1}=\begin{pmatrix}
			\bm{e}_{1}^T&
			\bm{e}_{1}^T
		\end{pmatrix}^T
		=\begin{pmatrix}
			0&1&0&1
		\end{pmatrix}^T;	
$$
$$
\m_{1,0}=\begin{pmatrix}
	\bm{e}_{0}^T&
	\bm{e}_{1}^T
\end{pmatrix}^T
=\begin{pmatrix}
	1&0&0&1
\end{pmatrix}^T,
$$
$$
\m_{1,1}=\begin{pmatrix}
	\bm{e}_{1}^T&
	\bm{e}_{0}^T
\end{pmatrix}^T
=\begin{pmatrix}
	0&1&1&0
\end{pmatrix}^T;
$$
$$
\m_{\infty,0}=\bm{e}_{0}\otimes \bm{1}_q=\begin{pmatrix}
	1&1&0&0
\end{pmatrix}^T,
$$
$$
\m_{\infty,1}=\bm{e}_{1}\otimes \bm{1}_q=\begin{pmatrix}
	0&0&1&1
\end{pmatrix}^T.
$$
In this case, the permuted Hadamard matrix $ \tilde{\bmH}^{(2)}$ equals $\bmH^{(2)}$. Two column vectors of the permuted Hadamard matrix $ \tilde{\bmH}^{(2)}$  are
	$$
	\tilde{\bm{h}}_0 =
	\begin{pmatrix}
		1\\
		1
	\end{pmatrix},
	\quad
	\tilde{\bm{h}}_1 =
	\begin{pmatrix}
		1\\
		-1
	\end{pmatrix}.
	$$
	It follows from \eqref{eq:mubss} that  MUBs  consists of
	\begin{align*}
		\bmB_0 
		=\frac{1}{\sqrt{2}}
		\left(
		\tilde{\bm{h}}_0\uparrow\m_{0,0}\ \
		\tilde{\bm{h}}_1\uparrow\m_{0,0}\ \ 
		\tilde{\bm{h}}_0\uparrow\m_{0,1}\ \ 
		\tilde{\bm{h}}_1\uparrow\m_{0,1}
		\right)
		=\frac{1}{\sqrt{2}}\begin{pmatrix}
			1 & 1 & 0 & 0\\
			0 & 0 & 1 & 1\\
			1 & -1 & 0 & 0\\
			0 & 0 & 1 & -1
		\end{pmatrix},
		\end{align*}
	\begin{align*}
		\bmB_1 
		=\frac{1}{\sqrt{2}}
		\left(\tilde{\bm{h}}_0\uparrow\m_{1,0}\ \ 
			\tilde{\bm{h}}_1\uparrow\m_{1,0}\ \ 
			\tilde{\bm{h}}_0\uparrow\m_{1,1}\ \ 
			\tilde{\bm{h}}_1\uparrow\m_{1,1}\right)
		=\frac{1}{\sqrt{2}}\begin{pmatrix}
			1 & 1 & 0 & 0\\
			0 & 0 & 1 & 1\\
			0 & 0 & 1 & -1\\
			1 & -1 & 0 & 0
		\end{pmatrix},
	\end{align*}
		and
		\begin{align*}
		\bmB_\infty 
		=\frac{1}{\sqrt{2}}
		\left(\tilde{\bm{h}}_0\uparrow\m_{\infty,0}\ 
			\tilde{\bm{h}}_1\uparrow\m_{\infty,0}\ 
			\tilde{\bm{h}}_0\uparrow\m_{\infty,1}\ 
			\tilde{\bm{h}}_1\uparrow\m_{\infty,1}\right)=\frac{1}{\sqrt{2}}\begin{pmatrix}
			1 & 1 & 0 & 0  \\
			1 & -1 & 0 & 0 \\
			0 & 0 & 1 & 1\\
			0 & 0 & 1 & -1
		\end{pmatrix}.
	\end{align*}
	We remark that the above MUBs has been constructed in \cite{Wocjan2004}.
	The dictionary $\bm{D}$ (see \eqref{eq:dic}) is given by
	\begin{align*}
		\sqrt{2}\bm{D} =\begin{pmatrix}
				1 & 1  & 0 & 0  &	1 & 1  & 0 & 0  & 	1 & 1  & 0 & 0\\
				0 & 0  & 1 & 1  & 	0 & 0  & 1 & 1  & 	1 & -1 & 0 & 0\\
				1 & -1 & 0 & 0  & 	0 & 0  & 1 & -1 & 	0 & 0  & 1 & 1\\
				0 & 0  & 1 & -1 & 	1 & -1 & 0 & 0  &  	0 & 0  & 1 & -1
			\end{pmatrix}.
	\end{align*}
	By \eqref{eq:x^i} and \eqref{eq:x}, the three sparse vector is given by
	$$\x = \begin{pmatrix}
		\x^0 \\ \x^1\\ \x^\infty
	\end{pmatrix} \in \mathbb{R}^{12}$$
	where three block vectors are
	\begin{align*}
		\x^0&=\bm{e}_{g^{(2)}_{0,0}}\otimes\bm{e}_0
		=\begin{pmatrix}
			1 & 0 & 0 & 0
		\end{pmatrix}^T,
		\\
		\x^1&=\bm{e}_{g^{(2)}_{1,1}}\otimes\bm{e}_1=\begin{pmatrix}
			0 & 0 & 0 & 1
		\end{pmatrix}^T,
		\\
		\x^\infty&=-\bm{e}_0\otimes\bm{e}_0=
		\begin{pmatrix}
			-1 & 0 & 0 & 0
		\end{pmatrix}^T.
	\end{align*}
	By Theorem \ref{main thm}, one can check such three sparse vector $\x$ is in the null space of $\bm{D}$, i.e.,
	$$
	\sqrt{2}\bm{D}\x =\bm{0}_4.
	$$
	Hence,
	$$
	\eta(\bm{D}) = 1 + \frac{1}{\mu(\bm{D})}=3,
	$$
	and
	$$
	\mu(\bm{D}) \eta(\bm{D}) = \frac{1}{2} \times 3 =  \left. \left(1 + \frac{1}{q}\right)\right|_{q=2}.
	$$

		\begin{figure}
		\centering
		\includegraphics[width=0.2\linewidth]{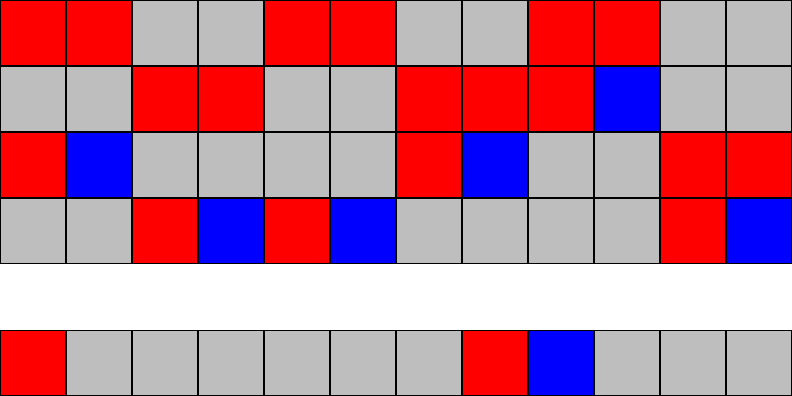}
		\caption{The dictionary  and the sparse vector in Example \ref{exa:2}. Blue: $-1$, Red: $1$, Gray: $0$}
		\label{fig:d1}
	\end{figure}

\end{Exa}

\begin{Exa}\label{exa:4}
Let $ q=2^2 $.  By the common construction, the Galois field $\F_{2^2}$ of order $4$ is $\F_2[x]/(x^2+x+1)$, and consists of the cosets of elements $0,1,x,x^2$. For convenience, we write $0,1,2$ and $3$ for the cosets of $0,1,x$ and $x^2$ respectively in this example, that is,
	$$
	\F_{2^2}=\{0,1,2,3\}.
	$$
The multiplication table and addition table for $\F_{2^2}$ are shown in Table \ref{Tab2M} and Table \ref{Tab2A}.
	\begin{table}[H]
			\centering
			\begin{tabular}{c|cccc}
				$\cdot$&0&1&2&3 \\
				\hline
				0&0&0&0&0\\
				1&0&1&2&3\\
				2&0&2&3&1\\
				3&0&3&1&2\\	
			\end{tabular}
		\caption{Multiplication  rule for $\F_{2^2}$}\label{Tab2M}
	\end{table}
	\begin{table}[H]
		\centering
		\begin{tabular}{c|cccc}
			$+$&0&1&2&3 \\
			\hline
			0&0&1&2&3\\
			1&1&0&3&2\\
			2&2&3&0&1\\
			3&3&2&1&0\\
		\end{tabular}
	\caption{Addition rule for $\F_{2^2}$}\label{Tab2A}
\end{table}
The matrix $\bm{G}^{(4)}$ defined in \eqref{matrixG} is
	$$
	\bm{G}^{(4)} 
	=\begin{pmatrix}
		0&0&0&0\\
		0&1&2&3\\
		0&2&3&1\\
		0&3&1&2\\
	\end{pmatrix}.
	$$
	Then the Latin squares (see \eqref{latin}) are
	$$
	\bmL^0=\begin{pmatrix}
		0&	1&	2&	3\\
		0&	1&	2&	3\\
		0&	1&	2&	3\\
		0&	1&	2&	3
	\end{pmatrix},
\quad
	\bmL^1=\begin{pmatrix}
		0&	1&	2&	3\\
		1&	0&	3&	2\\
		2&	3&	0&	1\\
		3&	2&	1&	0
	\end{pmatrix},
$$
$$
	\bmL^2=\begin{pmatrix}
		0&	1&	2&	3\\
		2&	3&	0&	1\\
		3&	2&	1&	0\\
		1&	0&	3&	2
	\end{pmatrix},
\quad
	\bmL^3=\begin{pmatrix}
		0&	1&	2&	3\\
		3&	2&	1&	0\\
		1&	0&	3&	2\\
		2&	3&	0&	1
	\end{pmatrix}.
	$$
	The matrix $\bm{A}^{(4)}$ defined in \eqref{A_L} (or \eqref{eq:L and A}) is
	$$
	\bmA^{(4)}=	
	\begin{pmatrix}
		0&1&3&2\\
		0&0&1&1\\
		0&3&0&3\\
		0&2&2&0
	\end{pmatrix}.
	$$
Analogue to the conversion between decimal and binary system, we write
$$
0=00,\quad 1=01,\quad 2=10,\quad 3=11,
$$
for the elements of $\F_{2^2}$ with respect to some basis of the vector space $\F_{2^2}$ over $\F_2$. Then the Hadamard matrix $ \bmH^{(4)} $ and the permuted  Hadamard matrix $ \tilde{\bmH}^{(4)} $ satisfy the conditions in Theorem \ref{a2hadamard} are
	$$
	\bmH^{(4)}= \bordermatrix{
	& 00 & 01 & 10 & 11 \cr
	00 & 1 & 1 &1 & 1 \cr
	01 & 1 & -1 &1 & -1 \cr
	10 & 1 & 1 &-1 & -1 \cr
	11 & 1 & -1 &-1 & 1 \cr
}
$$
and
$$
	\tilde{\bmH}^{(4)}=\bordermatrix{
	& 00 & 01 & 10 & 11 \cr
	00&	1 & 1 &1 & 1\cr
	01&	1 & -1 &-1 & 1\cr
	10&	1 & 1 &-1 & -1\cr
	11&	1 & -1 &1 & -1\cr
	}.
	$$
	Consequently, we can obtain the the orthonormal  bases
	by \eqref{ks-net}, \eqref{q-th block} and \eqref{eq:mubss},
	\begin{align*}
		\bmB_b=\frac{1}{2}
			\left(
			\tilde{\bm{h}}_0\uparrow\m_{b,0},
			\cdots,
			\tilde{\bm{h}}_3\uparrow\m_{b,0},\right.
			\left.\cdots,
			\tilde{\bm{h}}_0\uparrow\m_{b,3},	
			\cdots,
			\tilde{\bm{h}}_3\uparrow\m_{b,3}
			\right)
	\end{align*}
	where $ b=0,1,2,3,\infty.$
	The dictionary $\bm{D}$ (see \eqref{eq:dic}) is
	$$
	\bm{D}= \begin{pmatrix}
		\bm{B}_0 & \bm{B}_1 & \bm{B}_2 & \bm{B}_3 & \bm{B}_\infty
	\end{pmatrix} \in \mathbb{R}^{4^2 \times (4^2\times (4+1))}.
	$$
	Following \eqref{eq:x^i} and \eqref{eq:x}, five block vectors of $\x$ are defined as follows
	\begin{align*}
		\x^0&=\bm{e}_{0}\otimes\bm{e}_0
		, \\
		\x^1&=\bm{e}_{1}\otimes\bm{e}_1
		, \\
		\x^2&=\bm{e}_{3}\otimes\bm{e}_2
		, \\
		\x^3&=\bm{e}_{2}\otimes\bm{e}_3
		, \\
		\x^\infty&=-\bm{e}_0\otimes\bm{e}_0.
	\end{align*}
	By Theorem \ref{main thm} or direct calculation, we have
	$$2\bm{D}\x=\bm{0}_{16}.$$
	Hence,
	$$
	\left.\eta(\bm{D}) \right|_{q=2^2} =  1 + \left.\frac{1}{\mu(\bm{D})}\right|_{q=2^2}=5,
	$$
	and
	$$
	\mu(\bm{D}) \eta(\bm{D}) = \frac{1}{4}\times  5 =  \left.\left(1 + \frac{1}{q}\right)\right|_{q=2^2}.
	$$

	\begin{figure}
		\centering
		\includegraphics[width=0.85\linewidth]{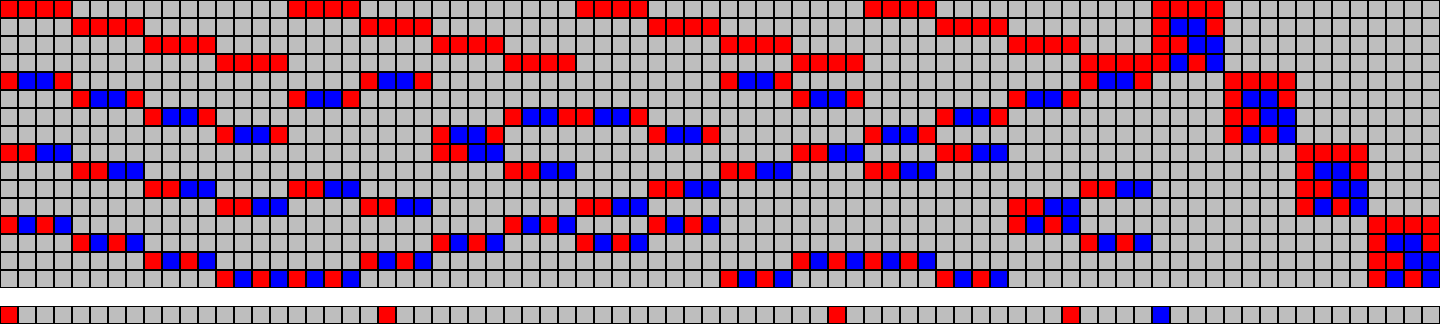}
		\caption{The dictionary  and the sparse vector in Example \ref{exa:4}. Blue: $-1$, Red: $1$, Gray: $0$}
		\label{fig:d2}
	\end{figure}
	
\end{Exa}

\begin{Exa}\label{Exa3}
We keep the notations in Example \ref{exa:2} and Example \ref{exa:4} and consider the field extension $\F_2\subseteq \F_{2^2}$ in this example. Clearly the vector space $\F_{2^2}/\F_2=\{[0],[2]\}$, where $[0]=\{0,1\}$ and $[2]=\{2,3\}$. We define an isomorphism of vector spaces over $\F_2$ as follows,
$$
\begin{array}{cclc}
\xi:  & \F_2 &\to  &\F_{2^2}/\F_2\\
~& 0   &\mapsto & [0],\\
~& 1 & \mapsto & [2].
\end{array}
$$
To our end, we choose an appropriate basis of the vector space $\F_{2^2}$ over $\F_2$ such that the elements of $\F_{2^2}$ has the following representations
$$
0=00,\quad 1=10,\quad 2=01,\quad 3=11.
$$
In this case, the Hadamard matrix $\bmH^{(4)}$ and the permuted Hadamard matrix $\tilde{\bmH}^{(4)}$ are
\[
	\bmH^{(4)}= \bordermatrix{
	& 00 & 10 & 01 & 11 \cr
	00 & 1 & 1 &1 & 1 \cr
	10 & 1 & -1 &1 & -1 \cr
	01 & 1 & 1 &-1 & -1 \cr
	11 & 1 & -1 &-1 & 1 \cr
}
\]
and
\[
	\tilde{\bmH}^{(4)}=\bordermatrix{
	& 00 & 10 & 01 & 11 \cr
	00&	1 & 1 &1 & 1\cr
	10&	1 & -1 &1 & -1\cr
	01&	1 & -1 &-1 & 1\cr
	11&	1 & 1 &-1 & -1\cr
	}.
\]
Then we define a map
$$
\begin{array}{cclc}
\iota:  & \F_2 &\to  &\F_{2^2}\\
~& 0   &\mapsto & 00,\\
~& 1 & \mapsto & 01.
\end{array}
$$
Consequently, we can obtain the  orthonormal  bases,
\begin{align*}
	\bmB_b=\frac{1}{2}
	\left(
	\tilde{\bm{h}}_0\uparrow\m_{b,0},
	\cdots,
	\tilde{\bm{h}}_3\uparrow\m_{b,0},\right.
	\left.
	\cdots,
	\tilde{\bm{h}}_0\uparrow\m_{b,3},
	\cdots,
	\tilde{\bm{h}}_3\uparrow\m_{b,3}
	\right)	
\end{align*}
	where $ b=0,1,2,3,\infty.$
	Differing from Example \ref{exa:4}, the  dictionary \(\bmD\) is
		\begin{equation}\label{exa:3}
			\bm{D}= \begin{pmatrix}
			\bm{B}_0 & \bm{B}_1 & \bm{B}_\infty
		\end{pmatrix}
	\in \mathbb{R}^{(2^2)^2 \times ((2^2)^2\times (2+1))}.
	\end{equation}
	Following \eqref{eq:y^i} and \eqref{eq:y}, three block vectors of $\y$ are defined by
	\begin{align*}
		\y^0&=(\bm{e}_{0}+\bm{e}_{1})\otimes\bm{e}_0
		, \\
		\y^1&=(\bm{e}_{2}+\bm{e}_{3})\otimes\bm{e}_2
		, \\
		\y^{\infty}&=-(\bm{e}_{0}+\bm{e}_{1})\otimes\bm{e}_0.
	\end{align*}
By Theorem \ref{main thm2},
$$2\bm{D}\y=\bm{0}_{16}.$$
We see that the dictionary defined in \eqref{exa:3} satisfies
	$$
	\eta(\bm{D}) =6 >1 + \frac{1}{\mu(\bm{D})}=5,
	$$
	and
	$$
	\mu(\bm{D}) \eta(\bm{D}) = \frac{1}{4}\times  6 =  \left.\left(1 + \frac{1}{q}\right)\right|_{q=2}.
	$$
\end{Exa}

\begin{figure}
	\centering
	\includegraphics[width=0.48\linewidth]{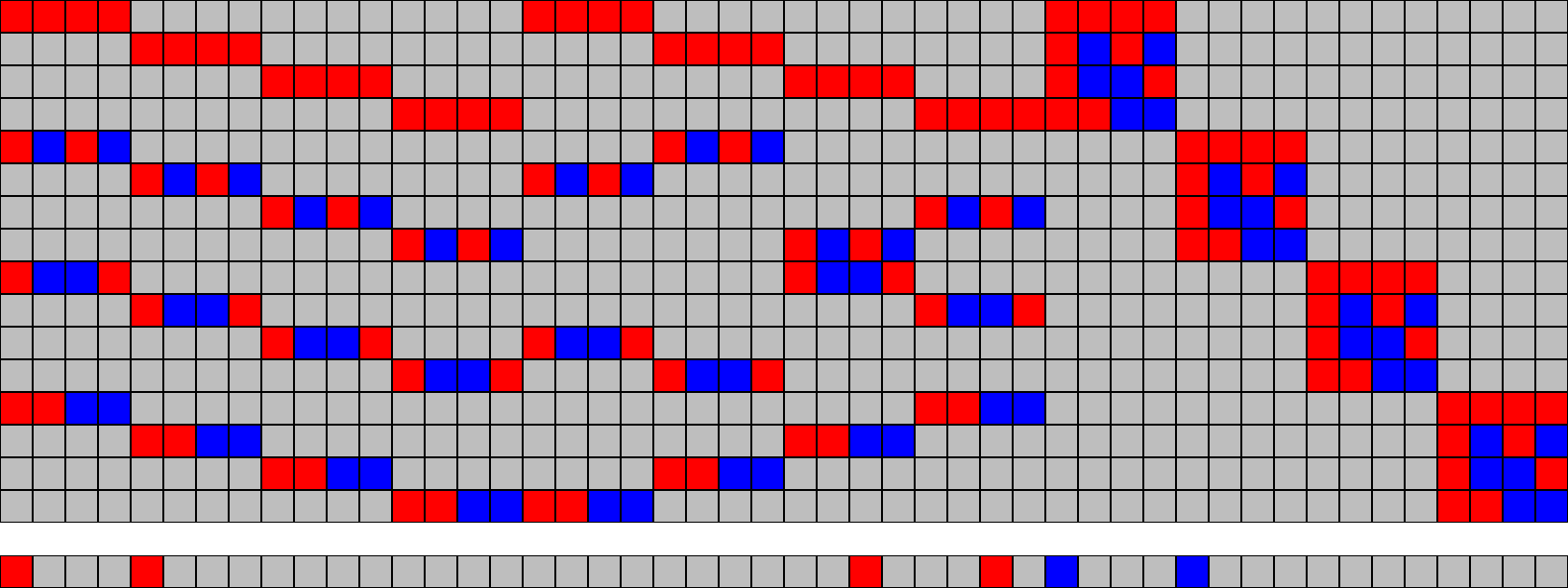}
	\caption{The dictionary  and the sparse vector in Example \ref{Exa3}. Blue: $-1$, Red: $1$, Gray: $0$}
	\label{fig:d333}
\end{figure}

\section{Conclusion}\label{sec:con}
In this paper, we  have studied the tightness of two estimates in  literatures on sparse representations.
Two classes of redundant dictionaries that are unions of several orthonormal bases were constructed by using the
mutually unbiased bases. To accurately calculate the sparks of such dictionaries,
the clear structures of sparse vectors in the null space are necessary.
Our results imply that two well-known estimates for the spark are indeed tight.
Therefore, Gribonval and Nielsen's open problem is answered positively.
The proofs of the above results are based on the techniques
in discrete mathematics and quantum information theory.
Further generalization of Theorem \ref{thm:L02} could be considered. For example,
the mutual coherence is $1/2^{q_1}$
and
the number of orthonormal bases  is $2^{q_2}+1$,
where $q_2$ is a factor of $q_1$.


%

%
%

\section*{Acknowledgment}

The authors would like to thank Tao Feng and Minyi Huang for helpful discussions regarding related material.

This work is supported by
the Zhejiang Provincial Natural Science Foundation of China under Grant No. LR19A010001,
the NSF of China under grant 12022112,12071426,11671358.

\ifCLASSOPTIONcaptionsoff
  \newpage
\fi

\bibliographystyle{IEEEtran}
\bibliography{IEEEabrv,references}

\end{document}